\documentclass[a4paper]{article}

\usepackage{url} %?
\usepackage{times} %?
\usepackage{verbatim}
\usepackage[T1]{fontenc}
\usepackage[swedish,english]{babel}
\usepackage{amsmath}
\usepackage{amssymb}
\usepackage{mathabx}
\usepackage[dvipdf]{graphicx}
\usepackage{color}
\usepackage{url}

\usepackage{epsfig}
\usepackage{float}
\usepackage[small]{caption}
\usepackage{subfig}
\usepackage{eurosym}
\graphicspath{{.}}

\def\reals{\mathbb{R}}

\def\logit{\mathop{\rm logit}\nolimits}
\def\diag{\mathop{\rm diag}}
\def\bin{\mathop{\rm Bin}}
\def\R{\mathbb{R}}

\def\half{\frac{1}{2}}

\newtheorem{theorem}{Theorem}
\newtheorem{proposition}[theorem]{Proposition}

\newenvironment{proof}
{\begin{trivlist}\item[\, 
{\bf Proof.}]}{{\hfill $\square$}\end{trivlist}}

\DeclareMathOperator*{\argmax}{\arg\!\max\ }

\title{A hidden Markov approach to\\ disability insurance}

\author{
Boualem Djehiche\footnote{Department of Mathematics, KTH Royal Institute of Technology, Sweden, {boualem@kth.se}. }
\and Björn Löfdahl \footnote{Department of Mathematics, KTH Royal Institute of Technology, Sweden, {bjornlg@kth.se}. }
}

\date{\today}

\begin{document}

\maketitle

\begin{abstract}
%A popular approach to estimating population transition rates, including mortality and disability, is the Lee-Carter model \cite{leecarter1992}. It relies on a two-step method, where in the first step, the age and time components of the mortality rates, namely the parameters $\alpha_x$, $\beta_x$ and $\kappa_t$, are estimated from historical data. Next, the time component $\kappa$ is assumed to follow a time series model, whose parameters are fitted to the estimated values of $\{\kappa_t\}$. This incoherent treatment of $\kappa$ may lead to both conceptual and numerical problems. In particular, the volatility of $\kappa$ tends to be overestimated, which may have significant impact on pricing and risk management of insurance products. We suggest that this general two-step approach can be improved in the following way: First, we assume a stochastic process form for the time component $\kappa$. The corresponding transition densities are incorporated into the likelihood, and the model parameters are estimated using the Expectation-Maximization algorithm. We illustrate the modelling procedure by fitting the model to disability claims data from the Swedish insurance company Folksam.
Point and interval estimation of future disability inception and recovery rates are predominantly carried out by combining generalized linear models (GLM) with time series forecasting techniques into a two-step method involving parameter estimation from historical data and subsequent calibration of a time series model. This approach may in fact lead to both conceptual and numerical problems since any time trend components of the model are incoherently treated as both model parameters and realizations of a stochastic process. We suggest that this general two-step approach can be improved in the following way: First, we assume a stochastic process form for the time trend component. The corresponding transition densities are then incorporated into the likelihood, and the model parameters are estimated using the Expectation-Maximization algorithm. We illustrate the modelling procedure by fitting the model to Swedish disability claims data.
\end{abstract}

\vspace{4mm}

\section{Introduction}

To determine premiums and reserves associated with health and disability insurance policies, the insurer needs predictions of the future rates of disability inception and recovery. While earlier research provides a solid base for disability modelling, new studies are required as the field is in constant change due to policy reforms and amendments to the existing regulations. For instance, the Swedish government launched major reforms of the national sickness insurance system in 2008, changing the rules for obtaining benefits from the Social Insurance Agency. This reform has been of major importance to the reduction in sickness absence. As noted by Aro {\it et al.} \cite{ADL13}, research of sickness and disability on Swedish data undertaken before 2008 may no longer provide an accurate description of the disability dynamics after the reform. As of October 2014, the Swedish government has suggested that the reforms of 2008 should essentially be reversed, a change that would require a drastic increase in premiums and reserves. This proposal highlights the need to study the calendar time dynamics of disability.

A popular approach to estimating disability inception and recovery rates is based on the generalized linear models framework. Renshaw and Haberman \cite{RenshawHaberman2000} model recovery, mortality and inception time trends in permanent health insurance using Poisson regression. Christiansen {\it et al.} \cite{Christiansen} model recovery, mortality and inception using the functional data approach of Hyndman and Ullah \cite{HyndmanUllah}. Aro {\it et al.} \cite{ADL13} propose logistic regression models for disability inception and termination. Similar methods are used for modelling population mortality, where perhaps the most well known is the Lee-Carter model \cite{leecarter1992} and its extensions, including the Poisson log-bilinear model of Brouhns {\it et al.} \cite{brouhns2002poisson}.

For the purpose of obtaining point and interval estimates of future disability or mortality rates, it is customary to combine GLMs with time series forecasting techniques into the following two-step method: In the first step, the parameters of the GLM are estimated from historical data. In the second step, a time trend component $\nu_t$ is assumed to follow a time series model, where a popular choice is the random walk with drift, and the parameters of the model are fitted to the estimated values of $\{\nu_t\}$. Prediction or simulation of future transition rates are obtained by prediction or simulation from the time series model for $\nu_t$. This two-step approach provides an easy way of fitting the model to data and simulating future outcomes, and it has been employed by e.g. Brouhns {\it et al.} \cite{brouhns2002poisson}, Christiansen {\it et al.} \cite{Christiansen}, Djehiche and L\"ofdahl \cite{DjehicheLofdahl14} and others.

%A popular approach to estimating population transition rates, including mortality and disability, is the Lee-Carter model \cite{leecarter1992} and its extensions. It relies on a two-step method, where in the first step the age and time components of the mortality rates, namely $\alpha_x$, $\beta_x$ and $\kappa_t$, are estimated from historical data. In the second step, the time trend component $\kappa_t$ is assumed to follow a time series model, where a popular choice is the random walk with drift, and the parameters of the model are fitted to the estimated values of $\kappa_t$. Prediction or simulation of future transition rates are obtained by prediction or simulation from the time series model for $\kappa_t$. This two-step approach provides an easy way of fitting the model to data and simulating future outcomes, and it has also been employed by e.g. Brouhns {\it et al.} \cite{brouhns2002poisson}, Hyndman and Ullah \cite{HyndmanUllah}, Christiansen {\it et al.} \cite{Christiansen}, Hainaut \cite{hainaut2012}, Djehiche and L\"ofdahl \cite{DjehicheLofdahl14} and others. 

An issue with the two-step approach is that at first, $\{\nu_t\}$ are considered parameters to be estimated. After estimating them, the assumption is altered so that $\nu$ is treated as a stochastic process. This may lead to both conceptual and numerical problems. In particular, the volatility of $\nu$ tends to be overestimated, which may have significant impact on pricing and risk management of insurance products. Qualitatively, this stems from the fact that yearly variations in the parameter values are caused by variations in the underlying process (systematic variation) as well as variations in the underlying population (idiosyncratic variation). The two-step approach makes no distinction between idiosyncratic and systematic variations.

The incoherence of the two-step approach has previously been pointed out by Czado {\it et al.} \cite{czado2005bayesian}. They propose to avoid this deficiency by integrating both steps into a Bayesian model, where the yearly values of the process corresponding to $\nu$ are all treated as random variables with given prior densities. The model parameters are then estimated using the Gibbs sampler and Metropolis-Hastings algorithm.

We suggest an alternative solution to this problem that does not rely on Bayesian statistics: First, we assume a stochastic process form for the time component $\nu$. Then, we fit the model over all time periods simultaneously, incorporating the transition densities of $\nu$ into the likelihood. Maximization of the likelihood cannot, however, be carried out directly, since now we no longer consider $\{\nu_t\}$ to be parameters. Instead, $\nu$ is treated as an unobservable stochastic process, so that the model is formulated as a Hidden Markov model (HMM). We then proceed using the Expectation-Maximization (EM) algorithm of Dempster {\it et al.} \cite{dempster1977maximum}. To the best of our knowledge, this type of EM-algorithm has not been used before for the purpose of estimating disability rates. It is, however, well known in other fields, such as finance \cite{duffie2009}.

As a starting point for the improved model, we consider the disability model from Aro {\it et al.} \cite{ADL13}:
Let $E_{x,t}$ be the number of healthy individuals aged $x$ at the beginning of time period $t$ in a given disability insurance scheme. We denote by $D_{x,t}$ the number of  individuals falling ill amongst the $E_{x,t}$ insured healthy individuals during time interval $[t,t+1)$. The authors assume that the conditional distribution of $D_{x,t}$ given $E_{x,t}$ is binomial:
\begin{equation}\label{eq:bin0}
D_{x,t}\sim\bin(E_{x,t},p_{x,t}),
\end{equation}
where $p_{x,t}$ is the probability that an $x$-year-old individual randomly selected at $t$ falls ill during $[t,t+1)$. Further, they suggest to model the logistic disability inception probabilities by
\begin{equation}\label{eq:logit0}
\logit p_{x,t} =\sum_{i=1}^m\nu^i_t\phi^i(x),
\end{equation}
where $\phi^i$ are user-defined basis functions, and $\nu^i_t$ are risk factors to be estimated from data. The authors also propose a straightforward extension of this model to disability termination modelling. 

The historical values of the risk factors $\nu_t=(\nu^1_t,\ldots,\nu^m_t)$ can be easily obtained by maximum likelihood estimation as follows. Given the historical values of $D_{x,t}$ and $E_{x,t}$, the log-likelihood function for yearly values of $\nu_t$ can be written using \eqref{eq:bin0} and \eqref{eq:logit0} as
\begin{align}
l_t(\nu_t;D_{\cdot,t}) &= \sum_{x\in X}\Big[D_{x,t}\sum_{i=1}^m \nu_t^i\phi^i(x) - E_{x,t}\log\big(1+\exp\big\{{\sum_{i=1}^m \nu_t^i\phi^i(x)}\big\}\big)\Big].\label{ll}
\end{align}
\noindent
Aro and Pennanen \cite{AroPennanen} show that if the basis functions are linearly independent, the yearly log-likelihood $l_t(\cdot;D_{\cdot,t})$ is strictly concave. Hence, maximizing $l_t(\cdot;D_{\cdot,t})$ over $\nu_t\in\reals^m$ using numerical methods gives a unique estimate of the vector $\nu_t$ for each $t$.

In this paper, we propose instead to treat $\nu$ as a hidden Markov process with transition densities parameterized by $\theta$, say, and use the Expectation-Maximization algorithm as follows: Given a parameter estimate $\theta^k$, integrate the complete data log-likelihood $l(\theta;D_{\cdot,1:n},\nu_{1:n})$ with respect to the distribution of $\nu_{1:n}:=(\nu_1,\ldots,\nu_n)$ conditional on the observations $D_{\cdot,1:n}:= (D_{\cdot,1},\ldots,\ D_{\cdot,n})$, e.g. let
\begin{equation}
Q(\theta|\theta^k) = E^{\theta^k}[l(\theta;D_{\cdot,1:n},\nu_{1:n})|D_{\cdot,1:n}].
\end{equation}
\noindent
Next, we maximize $Q$ w.r.t. $\theta$ to obtain
\begin{equation}
\theta^{k+1} = \argmax_\theta Q(\theta|\theta^k).
\end{equation}
\noindent
Iterating over expectation and maximization steps, the output of the EM-algorithm is a sequence $\{\theta^k\}$ of parameter estimates. Under technical conditions that are usually hard to verify, the sequence $\{\theta^k\}$ will eventually converge to a stationary point $\theta^*$ with $L(\theta^*)=L^*$ being the corresponding stationary point of the log-likelihood function. If the likelihood function is also unimodal, $\{\theta^k\}$ will converge to $\theta^* = \argmax_\theta L(\theta)$. See Wu \cite{wu1983convergence} for details.

The outline of this paper is as follows. In Section \ref{sec:inception_multi}, we propose a model for disability inception rates and show how the model parameters can be estimated using the EM algorithm. Section \ref{sec:inception_swedish} illustrates the modelling procedure by fitting the model to disability claims data from the Swedish insurance company Folksam. In Section \ref{sec:termination}, we propose a version of the model for the estimation of disability termination rates. In Section \ref{sec:termination_swedish}, we fit the termination model to disability claims data from Folksam.

\section{Disability inception model}\label{sec:inception_multi}
Let $E_{x,t}$ be the number of healthy individuals aged $x$ at the beginning of time period $t$ in a given disability insurance scheme. We denote by $D_{x,t}$ the number of  individuals falling ill amongst the $E_{x,t}$ insured healthy individuals during time interval $[t,t+1)$. Further, let $\nu$ be an $m$-dimensional Brownian motion starting at $\nu_0$ with drift $\mu$ and Cholesky matrix $A$. This choice of $\nu$ corresponds directly to the frequently used ARIMA(0,1,0) random walk. We assume that the conditional distribution of $D_{x,t}$ given $E_{x,t}$ and $\nu_{t}$ is binomial:
\begin{equation}\label{eq:bin2}
D_{x,t}\sim\bin(E_{x,t},p_{x,t}),
\end{equation}
where $p_{x,t}$ given by
\begin{equation}
p_{x,t}:=\frac{1}{1+e^{-g(x,\nu_t)}},
\end{equation}
\noindent
is the probability that an individual randomly selected at $t$ falls ill during $[t,t+1)$. Here, the selection of $g:\R^+\times\R^{m}\mapsto\R$ is a model choice. We adopt the basis function approach from \cite{ADL13}, and choose a $g$ of the form
\begin{equation}\label{eq:logit2}
g(x,\nu_t) =\sum_{i=1}^m\nu^i_t\phi^i(x),
\end{equation}
where $\phi^i,\ i=1,\ldots,m$, are user-defined basis function. 

Now, assume that we observe $D_{x,1:n}:= (D_{x,1},\ldots,D_{x,n})$ and \\ $E_{x,1:n}:= (E_{x,1},\ldots,E_{x,n})$, for $x$ from a given set $X$ of ages. Let $\theta = (\mu, A, \nu_0)$. Then, the complete data log-likelihood is given by
\begin{align}\label{logl2}
l(\theta;D_{\cdot,1:n},\nu_{1:n}) &= \sum_{t=1}^n\Big[l_t(\nu_t;D_{\cdot,t})+ \log f_{\nu_t|\nu_{t-1}}(\theta) + c_t\Big],
\end{align}
\noindent
where $f$ is the density of $\nu_t$ given $\nu_{t-1}$, and $c_t$ is a constant. From the Brownian motion assumption, we have
\begin{align}\label{eq:nu_dens}
\log f_{\nu_t|\nu_{t-1}}(\theta) =& -(\nu_t-\nu_{t-1}-\mu)^T(AA^T)^{-1}(\nu_t-\nu_{t-1}-\mu)\nonumber\\
&- \frac{1}{2}\log(\det(AA^T)).
\end{align}
\noindent
This is a direct extension of the model from \cite{ADL13}, in that instead of fitting each time period separately, we consider all time periods simultaneously by summation of the log-likelihood over all time periods. In addition, we include a term corresponding to the density of $\nu_t$ given $\nu_{t-1}$.

The following proposition is useful for obtaining point estimates of and confidence intervals for $\nu_t$.

\begin{proposition}\label{prop:filter_concave}
The filter density functions $\phi_{\nu_t}(\cdot) = f_{\nu_t|D_{\cdot,1:t}}(\cdot)$ are log-concave on $\R^m$.
\end{proposition}
\begin{proof}
Using Bayes' theorem, the law of total probability and the Markov property of $\nu$, we can write $f_{\nu_t|D_{\cdot,1:t}}(x_t)$ as
\begin{equation}
f_{\nu_t|D_{\cdot,1:t}}(x_t)\ \propto\ p_{D_{\cdot,t}|\nu_t}(x_t)I(x_t),
\end{equation}
\noindent
where
\begin{equation}
I(x_t) = \int f_{\nu_t|\nu_{t-1}}(x_t-x_{t-1})\prod_{k=1}^{t-1}p_{D_{\cdot,k}|\nu_k}(x_k)f_{\nu_k|\nu_{k-1}}(x_k-x_{k-1})dx_{1:t-1}.
\end{equation}
\noindent
Since $\log p_{D_{\cdot,t}|\nu_t}(x_t)=l_t(x_t;D_t)$ is concave, it remains to show that $I(x_t)$ is log-concave. It is well known that the densities $f_{\nu_k|\nu_{k-1}}(x_k-x_{k-1})$ are log-concave on $\R^m\times \R^m$. Hence, the integrand $h(x_t,x_{1:t-1})$ defined by
\begin{equation}
h(x_t,x_{1:t-1}) = f_{\nu_t|\nu_{t-1}}(x_t-x_{t-1})\prod_{k=1}^{t-1}p_{D_{\cdot,k}|\nu_k}(x_k)f_{\nu_k|\nu_{k-1}}(x_k-x_{k-1})
\end{equation}
\noindent
is log-concave on $\R^m\times \R^{m(t-1)}$. From \cite[Corollary 2]{Feyel}, log-concavity of $h$ directly implies log-concavity of $I(x_t)$ on $\R^m$.
\end{proof}
\noindent
It follows from Proposition \ref{prop:filter_concave} that the filter distributions are unimodal. This is a type of identification attribute of the model: Estimating the historical values of $\nu$ using the filter densities $\phi_{\nu_t}$ admits identification of $\nu_t,\ t=1,\ldots,n,$ by their respective modes. Confidence intervals for historical values of $\nu_t$ can be obtained directly as quantiles of $\phi_{\nu_t}$.

Unfortunately, the filter distributions cannot be calculated directly. However, it is relatively easy to sample from them using particle filter methods, given that we have estimates of $\theta$. By sampling from the filter distributions we can then obtain an updated estimate of $\theta$ using the Expectation-Maximization algorithm, with $\nu$ treated as a hidden Markov process. The choice of $\theta^0$ is important since we have not been able to show that the log-likelihood function is unimodal. We suggest choosing $\theta^0$ by fitting the model from Aro {\it et al.} for $t=1,\ldots,n$ and estimating $\theta^0$ from the time series of estimated values of $\nu_t$. This procedure should yield a good start guess for the parameters $\nu_0$ and $\mu$, while the start guess for the standard deviations given by $\sqrt{\diag(AA^T)}$ should be overestimated.

Integrating the log-likelihood and discarding all terms that do not depend on $\theta$, we obtain
\begin{align}\label{eq:Q_tilde2}
Q(\theta|\theta^k) = \sum_{t=1}^n&\Big[-\frac{1}{2}E^{\theta^k}[(\nu_t-\nu_{t-1}-\mu)^T(AA^T)^{-1}(\nu_t-\nu_{t-1}-\mu)|D_{\cdot,1:n}]\nonumber\\
& - \frac{1}{2}\log(\det(AA^T)) \Big].
\end{align}
\noindent
In order to maximize  $Q$ w.r.t. $\theta$, we need to evaluate the conditional expectations appearing in \eqref{eq:Q_tilde2}. This is not a trivial problem, since it is required to determine the density of $\nu_t-\nu_{t-1}$ conditional on $D_{\cdot,1:n}$ using Bayes' theorem. However, there exist numerical techniques that allow for efficient evaluation of the expectations, including particle filter methods. Further, we need to write $Q$ on a form that allows for easy maximization. We start with the latter task.

\subsection{Maximization}\label{sec:EM}
Simple but tedious linear algebra yields the following expression for $Q$:
\begin{align}\label{eq:Q_tilde3}
Q(\theta|\theta^k) =& -\frac{n}{2}\log(\det(AA^T)) - \half tr((AA^T)^{-1}C^T),
\end{align}
\noindent
where
\begin{align}
C_{ij} =& S_{ij}-\mu_iS_j-\mu_jS_i+n\mu_i\mu_j +E_{ij}-\nu_0^iE_j - \nu_0^jE_i+\nu_0^i\nu_0^j\nonumber\\
 &- \mu_i(E_j-\nu_0^j)-\mu_j(E_i-\nu_0^i),\\
S_{ij} =& \sum_{t=2}^nE^{\theta^k}[(\nu_t^i-\nu_{t-1}^i)(\nu_t^j-\nu_{t-1}^j)|D_{\cdot,1:n}],\\
S_i =& \sum_{t=2}^nE^{\theta^k}[(\nu_t^i-\nu_{t-1}^i)|D_{\cdot,1:n}],\\
E_{ij} =& E^{\theta^k}[\nu_1^i\nu_1^j|D_{\cdot,1:n}],\\
E_i =& E^{\theta^k}[\nu_1^i|D_{\cdot,1:n}].
\end{align}
\noindent
Given $S_{ij},\ S_i,\ E_{ij}$ and $E_i$ for $i,j=1,\ldots,m$, it is a simple matter to maximize \eqref{eq:Q_tilde3} in the following way: First, taking derivatives, the optimal $\mu$ and $\nu_0$ are given by
\begin{align}
\mu_i   &= \frac{1}{n-1}S_i,\\
\nu_0^i &= E_i - \frac{1}{n-1}S_i.
\end{align}
\noindent
Substituting $\mu$ and $\nu_0$ back into $C_{ij}$ yields
\begin{equation}
C_{ij} = S_{ij} + E_{ij} -\frac{1}{n-1}S_iS_j -E_iE_j.
%(AA^T)_{ij} &= \frac{1}{n}C_{ij}.
\end{equation}
\noindent
Now, since C is no longer a function of $\theta$, it suffices to consider the mapping $A\mapsto\tilde Q(A)$ defined by
\begin{align}\label{eq:Q_tildeA}
\tilde Q(A) =& -\frac{n}{2}\log(\det(AA^T)) - \half tr((AA^T)^{-1}C^T)\nonumber\\
=& -\frac{n}{2}\log(\det(AA^T)) - \frac{n}{2} tr((AA^T)^{-1}\bar C^T),
\end{align}
\noindent
where $\bar C_{ij} = \frac{1}{n}C_{ij}$. It is well known that \eqref{eq:Q_tildeA} obtains its maximum value at
\begin{equation}
AA^T = \bar C,
\end{equation}
\noindent
provided that $\bar C$ is positive definite. Occasionally, due to Monte Carlo error, it may happen that $\bar C$ is not positive definite. This can be remedied in several ways, we may for example attempt to maximize \eqref{eq:Q_tildeA} numerically. Another option is to resample and perform the E-step anew, and attempt the M-step once more using the updated expectations.

\subsection{Expectation}\label{sec:EM2}
We now turn towards the task of computing the conditional expectations. The conditional expectations $S_{ij},\ S_i,\ E_{ij}$ and $E_i$ are of a form suitable to the particle-based rapid incremental smoother, or PaRIS, algorithm due to Westerborn and Olsson \cite{westerbornolsson2014}.

A particle filter is a necessary requirement for implementing the PaRIS algorithm, and for this purpose we choose to implement a simple bootstrap particle filter. The filter distributions $\phi_{\nu_t}$, that is, for each $t$, the distribution of $\nu_t$ conditional on $D_{\cdot,1:t}$, are estimated in the following way: Given a sample of $\nu_{t-1}$, we first sample $N$ particles of $\nu_t$ from $f_{\nu_t|\nu_{t-1}}(\theta^k)$ to obtain $z_k = (z_k^i)_{i=1,\ldots,m},\ k=1,\ldots,N$. Each particle $z_k$ is then given the weight $w_k \propto \exp\{l_t(z_k;D_{\cdot,1:t})\}$, and the filter probability mass function is estimated by $\hat\phi_{\nu_t}(z_k) = w_k$. Finally, we bootstrap from $z_k,\ k=1,\ldots,N$ with probabilities $w_k,\ k=1,\ldots,N$, to obtain a sample of $\nu_t$, and repeat the procedure until $t=n$. We estimate the yearly values of $\nu_t$ for $t=1,\ldots,n$ by
\begin{equation}\label{eq:nu_hat}
\widehat \nu_t^i = \sum_{k=1}^N w_k^iz_k^i,\ i=1,\ldots,m.
\end{equation}
\noindent
Confidence intervals for $\nu_t^i$, $i=1,\ldots,m$, are obtained by calculating the empirical quantiles based on $\hat\phi_{\nu_t}(\cdot)$.
Finally, the expectations $S_{ij},\ S_i,\ E_{ij}$ and $E_i$ are estimated using the PaRIS algorithm as outlined in \cite{westerbornolsson2014}.

Note that this separation of expectation and maximization is possible due to the model specification: Since the basis functions are chosen {\it a priori} and are not themselves to be estimated, all terms of $Q$ that depend on both $\theta$ and $\nu$ are product terms. No other non-multiplicative dependencies are present. This allows us to write $Q$ of the form \eqref{eq:Q_tilde3}, which allows for separating the expectation and maximization steps as required.

Consider the case where the model specification was written so that the basis functions were also to be estimated from data. For example, we may consider the Lee-Carter type model from \cite{brouhns2002poisson} for the force of mortality $q_{x,t}$:
\begin{equation}\label{eq:LC}
q_{x,t}:=e^{\alpha_x+\beta_x\kappa_t},
\end{equation}
\noindent
where $\alpha_x$ and $\beta_x$ are to be estimated along with $\kappa_t$. Then, $Q$ will contain terms $T_{x,t}$ of the form
\begin{equation}
T_{x,t}= E_{x,t}e^{\alpha_x}E^{\theta^k}[e^{\beta_x\kappa_t}|D_{\cdot,1:t}].
\end{equation}
\noindent
In our approach, the conditional expectation $E^{\theta^k}[e^{\beta_x\kappa_t}|D_{\cdot,1:t}]$ can only be estimated for a fixed $\beta_x$. Hence, we cannot feasibly implement this version of the EM-algorithm for the model specified by \eqref{eq:LC}, except by estimating this quantity over a range of values for $\beta_x$ and using interpolation and extrapolation over this range in the M-step. It is, however, possible to fit a model of the type \eqref{eq:logit2} to mortality data using the techniques of this Section. For a discussion on how to choose the basis functions $\{\phi^i\}$ for mortality modelling we refer to Aro and Pennanen \cite{AroPennanen}.

\section{Fitting Swedish disability inception rates}\label{sec:inception_swedish}
In this section, we implement the EM-algorithm from Sections \ref{sec:EM}-\ref{sec:EM2} for the disability inception model from Section \ref{sec:inception_multi}, and fit it to population data from Folksam.

\subsection{Two-factor model}
We implement the model from Section \ref{sec:inception_multi} with basis functions given by
\begin{equation*}
\phi^1(x) =  \frac{64-x}{39} \quad  \text{and} \quad \phi^2(x) = \frac{x-25}{39},
\end{equation*}
for $x \in [25,64]$. The linear combination $\sum_{i=1}^2 \nu^i_t\phi^i(x)$ is also linear. 
Note that the same linear form for the curve of logit $p_{\cdot,t}$ could have been obtained using any two linearly independent linear basis functions. However, this particular choice ensures a certain natural interpretation of the stochastic process $\nu$. Namely, for every $t$,
\begin{equation*}
\logit p_{25,t}=\nu^1_t\phi^1(25)+\nu^2_t\phi^2(25)=\nu^1_t,
\end{equation*}
and, similarly, $\logit p_{64,t}=v_t^2$. Hence, two components of $\nu$ represent the logit disability inception probabilities of ages $25$ and $64$, respectively. 

The EM-algorithm stabilizes to within Monte Carlo error after about 180 iterations. We run it for 20 more iterations and estimate $\theta$ as the average over the 20 last iterations. The value of $Q(\theta^k|\theta^{k-1})$ for $k=1,\ldots,\ 200$ is presented in Figure \ref{fig:inception_Q_2par}. The estimated inception probabilities from the Hidden Markov model (HMM) for the years 2000-2011 are displayed in Figures \ref{fig:inception_2par_params}-\ref{fig:inception_2par_mesh}. For reference, they are compared to the estimations from \cite{ADL13}, (hereafter referred to as the multi-period model). Note that, due to confidentiality, the actual values of the estimates are not reported. Figures \ref{fig:inception_2par_filter_1}-\ref{fig:inception_2par_filter_2} display the estimated filter densities for $\nu^1$ and $\nu^2$, respectively. Indeed, as inferred from Proposition \ref{prop:filter_concave}, the estimated filter distributions are for the most part unimodal. They also seem to be symmetric, which makes estimation of the yearly values of $\nu^1$ and $\nu^2$ from their corresponding mean values or modes equivalent. Table \ref{table:2par} displays the estimated drift and volatility parameters from the HMM as a fraction of the corresponding estimates from the multi-period model.

\begin{figure}[!ht]
\begin{center}
\epsfig{file=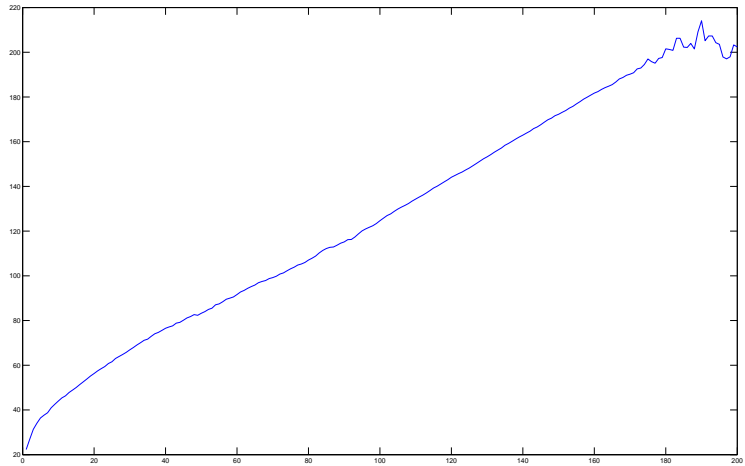,height=0.6\linewidth, width=\linewidth, angle=0}
\vspace{-25pt}
\caption{The value of $Q(\theta^k|\theta^{k-1})$ for $k=1,\ldots,\ 200$.}
\label{fig:inception_Q_2par}
\end{center}
\end{figure}

\begin{figure}[!ht]
\begin{center}
\epsfig{file=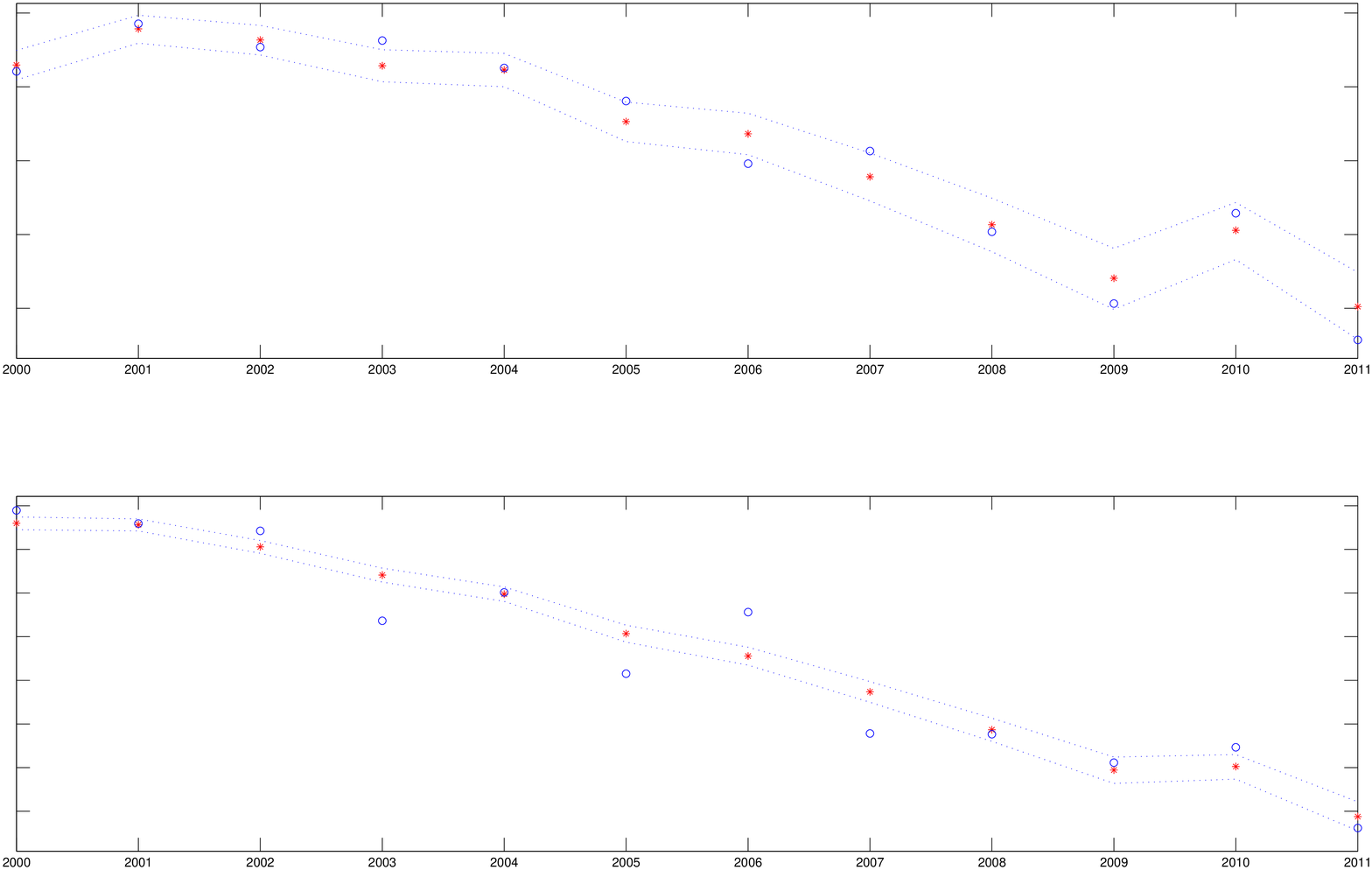,height=0.6\linewidth, width=\linewidth, angle=0}
\vspace{-25pt}
\caption{Estimates of $\nu_{1:n}$ (stars) with confidence bands (dashed). Estimates from \cite{ADL13} (circles) for comparison.}
\label{fig:inception_2par_params}
\end{center}
\end{figure}

\begin{figure}[!ht]
\begin{center}
\epsfig{file=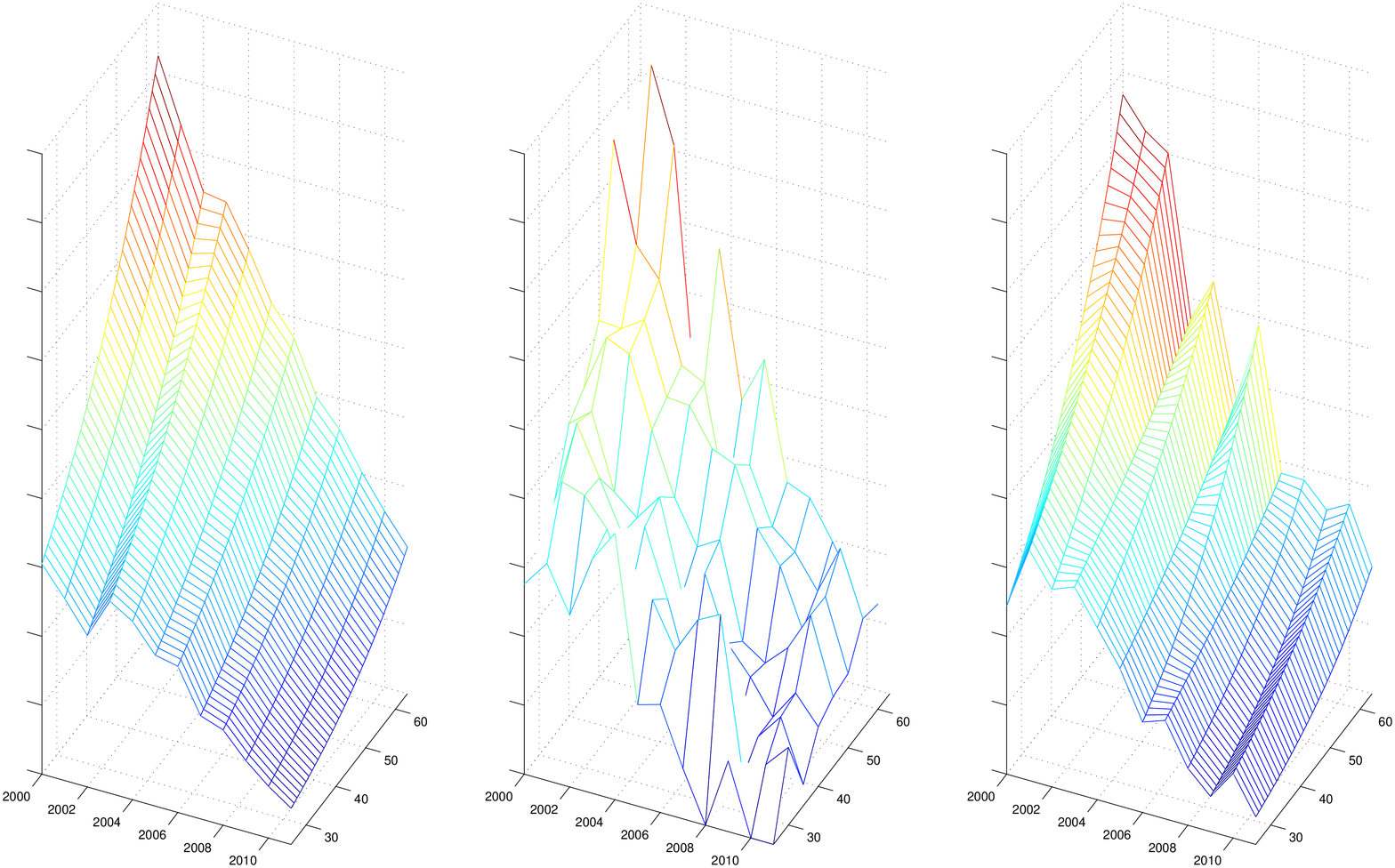,height=0.6\linewidth, width=\linewidth, angle=0}
\vspace{-25pt}
\caption{Left: Estimates of $p_{25:64,1:n}$. Middle: Raw data $D_{25:64,1:n}/E_{25:64,1:n}$. Right: Estimates from \cite{ADL13}.}
\label{fig:inception_2par_mesh}
\end{center}
\end{figure}

\begin{figure}[!ht]
\begin{center}
\epsfig{file=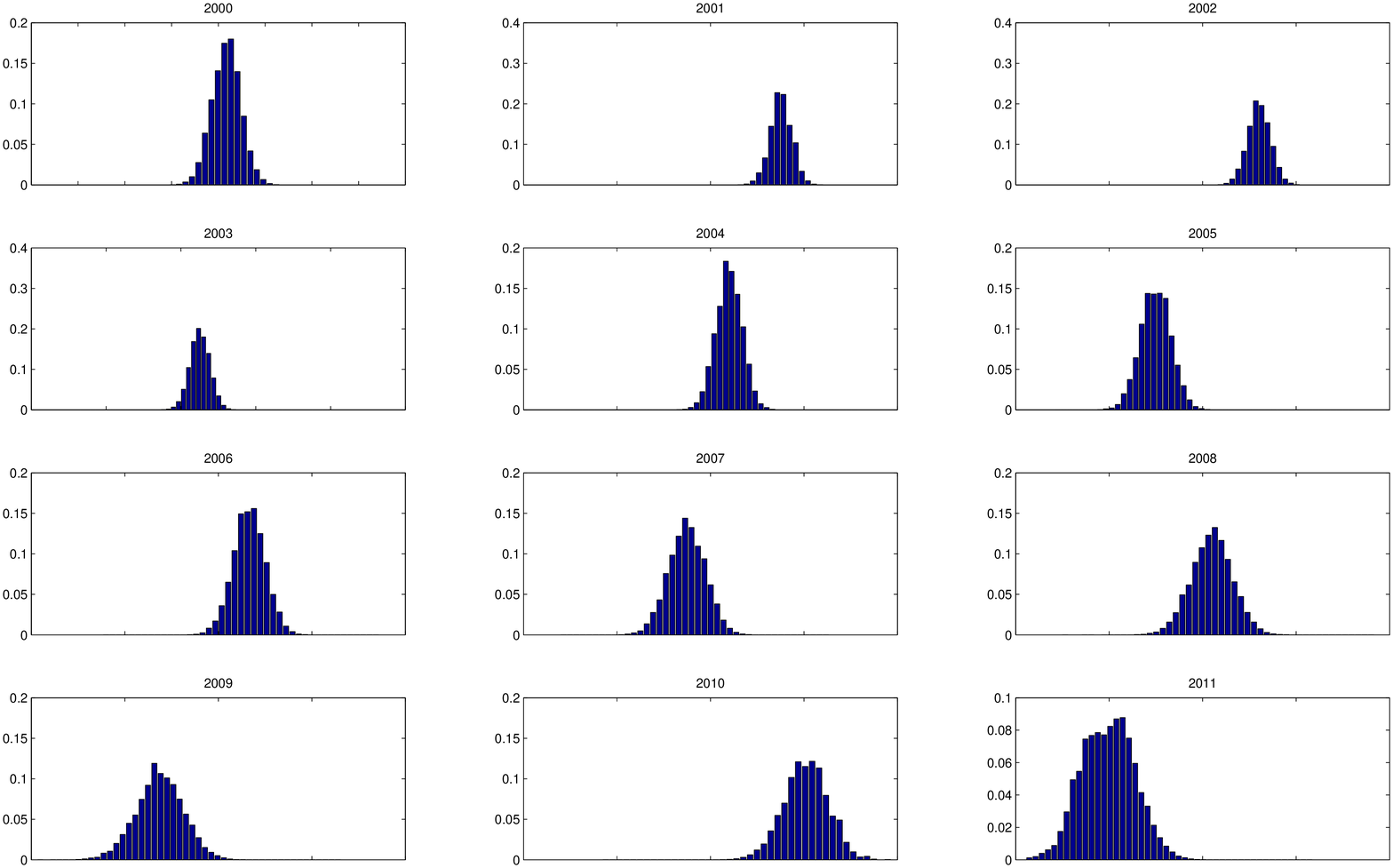,height=0.6\linewidth, width=\linewidth, angle=0}
\vspace{-25pt}
\caption{Estimated filter densities for $\nu^1$.}
\label{fig:inception_2par_filter_1}
\end{center}
\end{figure}

\begin{figure}[!ht]
\begin{center}
\epsfig{file=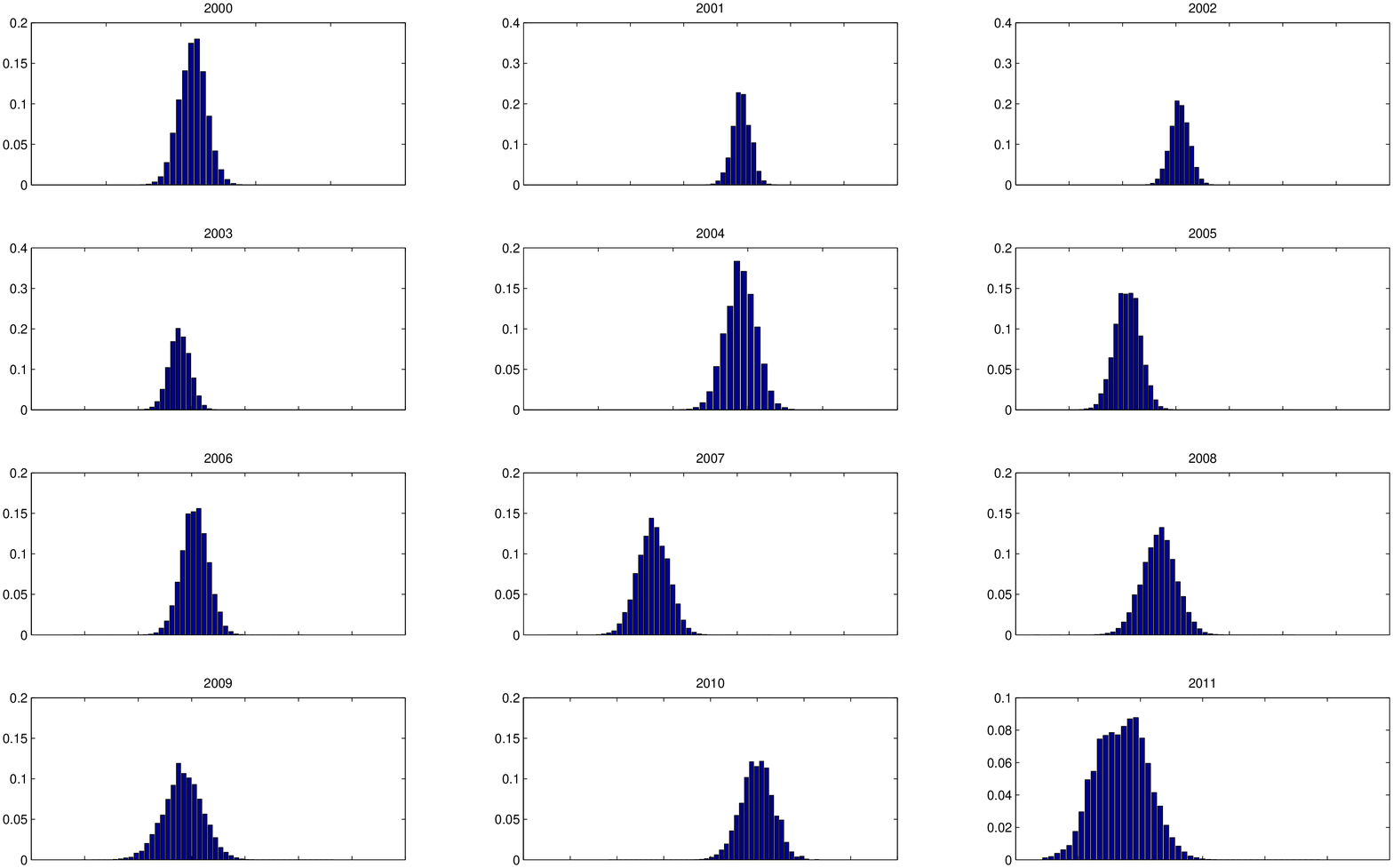,height=0.6\linewidth, width=\linewidth, angle=0}
\vspace{-25pt}
\caption{Estimated filter densities for $\nu^2$.}
\label{fig:inception_2par_filter_2}
\end{center}
\end{figure}

\begin{table}[!ht]
\caption{Relative difference of the estimated drift and volatility parameters between the two models.}
\label{table:2par}
\footnotesize
\begin{center}
\begin{tabular}{c|cc}
% $\Delta z = \Delta t$
 & $\mu$ & $\sigma$ \\
\hline
$\nu^1$ & 0.92 & 0.48\\
$\nu^2$ & 0.93 & 0.23\
\end{tabular} \\
\end{center}
%{\itshape{Notes.}} 
%Possible notes.
\end{table}
The HMM seems to provide estimates of $\nu_{1:n}$ and $\mu$ that are quite close to the estimates from the multi-period model, where the model is fitted to data for each time period separately. The goodness of fit as measured by the yearly log-likelihood values $l_t(\nu_t;D_{\cdot,t})$ is worse for the HMM, but this is to be expected: Calibration on only one time period will yield a better fit to the data corresponding to that particular time period compared to simultaneous calibration over many time periods, unless the estimates of $\nu_t$ turn out identical. The HMM provides smoothing across time due to the fact that the yearly estimates are essentially parameterized by a stochastic process. 

The purpose of calibrating a model is usually not to obtain the best possible fit to historical data. Rather, the objective should be to obtain the best fit of the law of future disability rates. The estimated diffusion coefficients from the HMM are considerably reduced compared to the multi-period model, which has a significant impact on the generation of future disability rates, prices, risks and capital charges. This is both reasonable expected. Qualitatively, this stems from the fact that yearly variations in parameter values are caused by variations in the underlying process (systematic variation) as well as variations in the underlying population (idiosyncratic variation). The multi-period model makes no distinction between idiosyncratic and systematic variations. The HMM, on the other hand, enables us to distinguish between the two, and in effect the yearly changes of $\nu$ are dampened by the stochastic process interpolation. 

As regards the number of free parameters to be estimated, the multi-period model yields estimates of $\nu_t^i,\ t=1,\ldots,\ n,\ i=1,\ldots,\ m$, for a total of $nm$ free parameters. The drift vector $\mu$ and cholesky matrix $A$ are then simply functions of the yearly parameter values. On the other hand, the HMM yields estimates of $\mu, A$ and $\nu_0$ for a total of $\frac{m^2}{2}+\frac{5m}{2}$ free parameters. The yearly estimates of $\nu_t^i,\ t=1,\ldots,\ n,\ i=1,\ldots,\ m$, are then given as integrals, parameterized by $\mu, A$ and $\nu_0$. Thus, the number of free parameters is lower for the HMM if $\frac{m^2}{2}+\frac{5m}{2} < nm$, or, equivalently, $m<2n-5$. For our sample of $n=12$ years, the HMM has fewer free parameters if $m<19$. Usually, we are content with a much lower dimension for the environment process $\nu$. 

Given these two major advantages, we conclude that, for the purpose of generating future scenarios, the HMM is preferred to the multi-period model. We will now try to refine the model by increasing the number of basis functions.

\subsection{Three-factor model}
Following \cite{ADL13}, we implement the model from Section \ref{sec:inception_multi} with piecewise linear basis functions given by
\begin{align*}
\phi^1(x) &= 
\begin{cases}
1-\frac{x-25}{15} & \text{for $x\in [25,40)$} \\
0 & \text{for $x\in [40,64]$},
\end{cases}\\
\phi^2(x) &= 
\begin{cases}
\frac{1}{15}(x-25) & \text{for $x\in [25,40)$} \\
\frac{64}{24}-\frac{x}{24} & \text{for $x\in [40,64]$},
\end{cases}\\
\phi^3(x) &= 
\begin{cases}
0 & \text{for $x\in [25,40)$}\\
\frac{x}{24}-\frac{40}{24} & \text{for $x\in [40,64]$}.
\end{cases}
\end{align*}
\noindent
The linear combination $\sum_{i=1}^3 \nu_t^i\phi^i(x)$ is now piecewise linear and continuous, with mid point at $x=40$ years. As in the two-factor model, the values of the factors are points on the logit inception probability curve: $\logit p_{25,t}=\nu^1_t$, $\logit p_{40,t}=\nu^2_t$ and $\logit p_{64,t}=\nu^3_t$.

The EM-algorithm stabilizes to within Monte Carlo error after about 250 iterations. The estimated inception probabilities are displayed in Figures \ref{fig:inception_3par_params}-\ref{fig:inception_3par_mesh}. For reference, they are compared to the estimations from the multi-period model. Table \ref{table:3par} displays the estimated drift and volatility parameters from the HMM as a fraction of the corresponding estimates from the multi-period model.
\begin{figure}[!ht]
\begin{center}
\epsfig{file=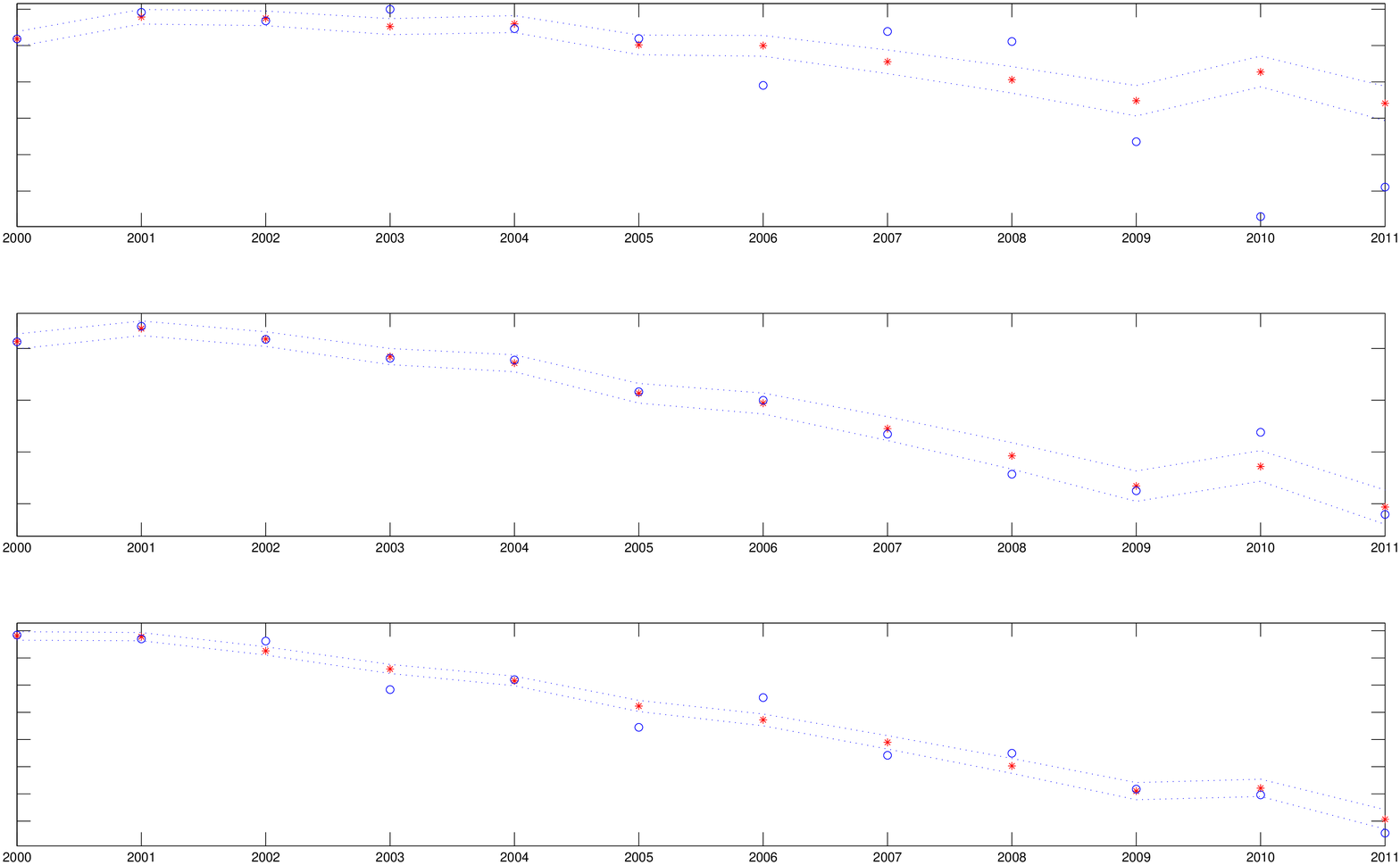,height=0.6\linewidth, width=\linewidth, angle=0}
\vspace{-25pt}
\caption{Estimates of $\nu_{1:n}$ (stars) with confidence bands (dashed). Estimates from \cite{ADL13} (circles) for comparison.}
\label{fig:inception_3par_params}
\end{center}
\end{figure}

\begin{figure}[!ht]
\begin{center}
\epsfig{file=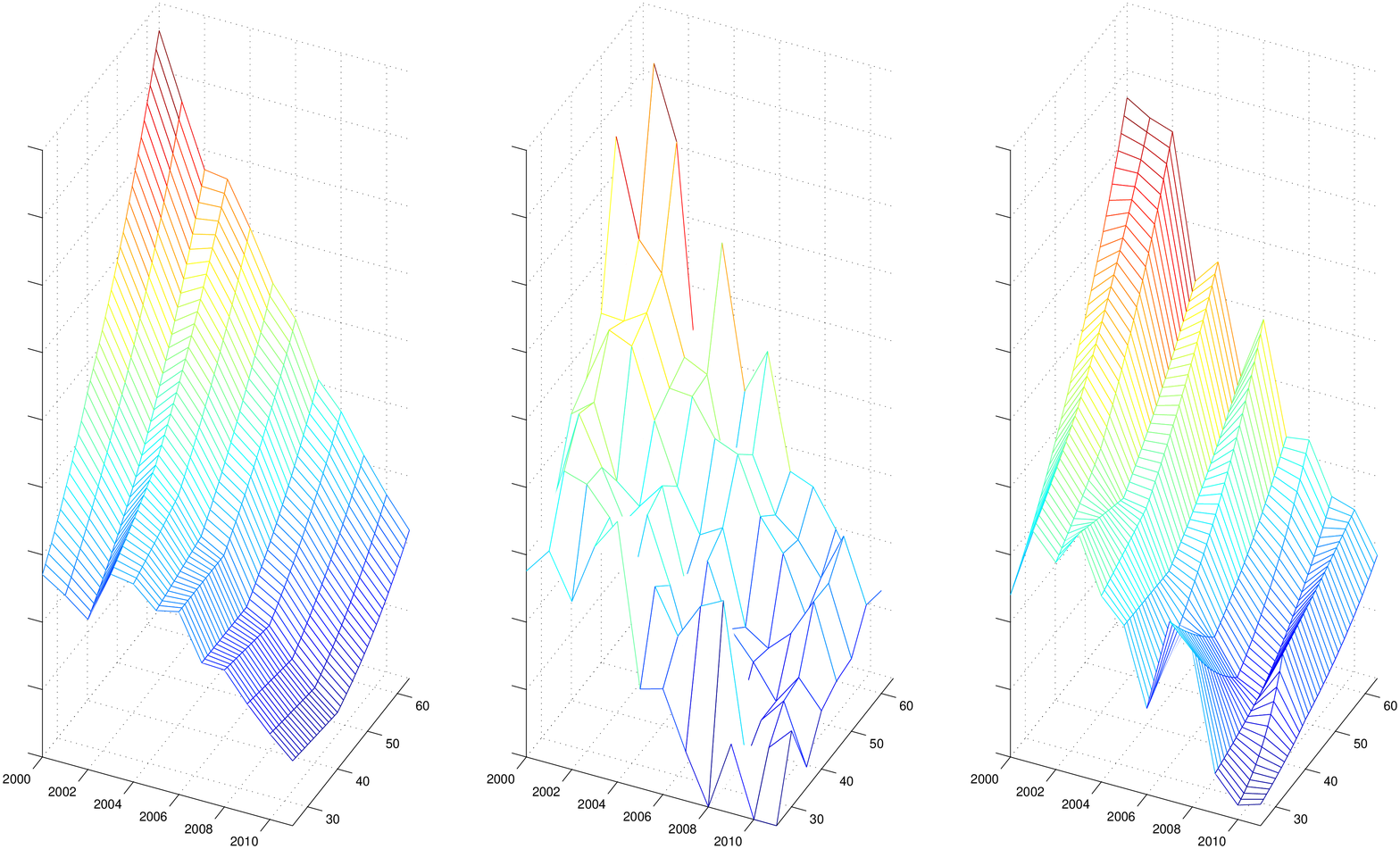,height=0.6\linewidth, width=\linewidth, angle=0}
\vspace{-25pt}
\caption{Left: Estimates of $p_{25:64,1:n}$. Middle: Raw data $D_{25:64,1:n}/E_{25:64,1:n}$. Right: Estimates from \cite{ADL13}.}
\label{fig:inception_3par_mesh}
\end{center}
\end{figure}

\begin{table}[!ht]
\caption{Relative difference of the estimated drift and volatility parameters between the two models.}
\label{table:3par}
\footnotesize
\begin{center}
\begin{tabular}{c|ccc}
% $\Delta z = \Delta t$
 & $\mu$ & $\sigma$ \\
\hline
$\nu^1$ & 0.45 & 0.32 \\
$\nu^2$ & 0.98 & 0.42 \\
$\nu^3$ & 0.93 & 0.29 \\
\end{tabular} \\
\end{center}
%{\itshape{Notes.}} 
%Possible notes.
\end{table}
The HMM seems to provide estimates of $\nu_{1:n}$ and $\mu$ that are quite close to the estimates from the multi-period model, at least for $\nu^2$ and $\nu^3$. For $\nu^1$, the differences are quite pronounced from the year 2006 and onwards. It may be that the Brownian motion assumption could be altered, but this is a topic for future research. We return to this discussion point later on. 

As was the case with the two-factor model, the estimated diffusion coefficients are considerably reduced. Judging from Figures \ref{fig:inception_3par_params}-\ref{fig:inception_3par_mesh}, adding more complexity in terms of another underlying factor does not seem to significantly enhance the model. Here, the results are based on a midpoint of $x=40$ years, but from extensive testing we find that the conclusion is valid for any choice of midpoint.

\section{Disability termination model}\label{sec:termination}
The algorithm from Section \ref{sec:inception_multi} can easily be applied to obtain an extension of the disability termination model from \cite{ADL13}.

Let $E_{x,d,t}$ be the number of individuals with disability inception age $x$ and disability duration $d$ at some point in the time period $[t,t+1)$. Further, let $R_{x,d,t}$ denote the number of individuals among $E_{x,d,t}$ with termination during $[t,t+1)$ and $[d,d+\Delta d)$. Further, let $\nu$ be a $k\times m$-dimensional Brownian motion starting at $\nu_0$ with drift vector $\mu$ and Cholesky matrix $A$. We assume that the conditional distribution of $R_{x,d,t}$ given $E_{x,d,t}$ and $\nu_{t}$ is binomial:
\begin{equation}\label{eq:bin_avv}
R_{x,d,t}\sim\bin(E_{x,d,t},p_{x,d,t}),
\end{equation}
where 
\begin{equation}
p_{x,d,t}:=\frac{1}{1+e^{-g(x,d,\nu_t)}}
\end{equation}
\noindent
denotes the probability that the disability of an individual, with disability inception age $x$ and disability duration $d$ at some point in the time period $[t,t+1)$, is terminated before duration $d+\Delta d$. Again, we adopt the basis function approach from \cite{ADL13}, and choose a function $g:\R^+\times\R^+\times\R^{k\times m}\mapsto \R$ of the form
\begin{equation}\label{eq:logit_avv}
g(x,d,\nu_t) =\sum_{i=1}^m \phi^i(x)\sum_{j=1}^k\psi^j(d)\nu^{i,j}_t,
\end{equation}
\noindent
where $\phi$ and $\psi$ are basis functions in $x$ and $d$, respectively.

Now, assume that we observe $R_{x,d,1:n}:= (R_{x,d,1},\ldots,R_{x,d,n})$, for $x$ and $d$ from given sets $X$ and $D$ of ages and disability durations, respectively. Let $\theta = (\mu, A, \nu_0)$. Then, the complete data log-likelihood is given by
\begin{align}\label{logl_avv}
l(\theta;R_{\cdot,\cdot,1:n},\nu_{1:n}) &= \sum_{t=1}^n\Big[\sum_{\substack{x\in X\\d\in D}}\Big\{R_{x,d,t}g(x,d,\nu_t) - E_{x,d,t}\log\big(1+\exp\big\{g(x,d,\nu_t)\big\}\big)\Big\}\nonumber\\
&+ \log f_{\nu_t|\nu_{t-1}}(\theta) + c_t\Big],
\end{align}
\noindent
where $f$ denotes the density of $\nu_t$ given $\nu_{t-1}$, and $c_t$ is a constant. Integrating the log-likelihood and discarding terms that do not depend on $\theta$, we again obtain
\begin{align}\label{eq:Q_tilde_avv}
Q(\theta|\theta^k) &= \sum_{t=1}^n\Big[-\frac{1}{2}E^{\theta^k}[(\nu_t-\nu_{t-1}-\mu)^T(AA^T)^{-1}(\nu_t-\nu_{t-1}-\mu)|R_{\cdot,\cdot,1:n}]\nonumber\\
& - \frac{1}{2}\log(\det(AA^T)) \Big].
\end{align}
\noindent
The expectation and maximization steps are carried out exactly as in Sections \ref{sec:EM}-\ref{sec:EM2}.

\section{Fitting Swedish termination rates}\label{sec:termination_swedish}
In this section, we implement the EM-algorithm from Sections \ref{sec:EM}-\ref{sec:EM2} for the disability termination model from Section \ref{sec:termination}, and fit it to population data from Folksam.

\subsection{Four-factor model}
We propose an initial model for the termination probabilities with linear basis functions in both the age dimension $x$ and the duration dimension $d$:
\begin{eqnarray*}
\phi^1(x) & = & \frac{64-x}{39},\\
\phi^2(x) & = & \frac{x-25}{39}, \\
\psi^1(d) & = & 1,\\
\psi^2(d) & = & d.\\
\end{eqnarray*}
\noindent
The obtained termination curve for a 25-year old is only affected by the term containing $\phi^1$. Notice that, since $\phi^1(25)=1$ and $\phi^2(25)=0$, $\logit p(25,d,\nu_t) = \nu^{1,1}_t + d\nu^{1,2}_t$. In other words, the processes $\nu^{1,1} $ and $\nu^{1,2}$ correspond to the termination curve for a 25-year old. By the same argument, the processes $\nu^{2,1}$ and $\nu^{2,2}$ correspond to the termination curve for a 64-year old. The logistic conditional probability of termination for an $x$ year old is thus a convex combination of the logistic conditional probabilities of termination for a 25-year old and a 64-year old.

The EM-algorithm stabilizes to within Monte Carlo error after about 260 iterations. The estimated parameters are displayed in Figure \ref{fig:termination_4par_params}. Figures \ref{fig:termination_4par_mesh_f2006}-\ref{fig:termination_4par_mesh_f2010} displays the termination surface from the model alongside the corresponding Kaplan-Meier curves for the years 2006 and 2010, respectively. Results from \cite{ADL13} are displayed for reference. Table \ref{table:4par} displays the estimated drift and volatility parameters from the HMM as a fraction of the corresponding estimates from the multi-period model.

\begin{figure}[!ht]
\begin{center}
\epsfig{file=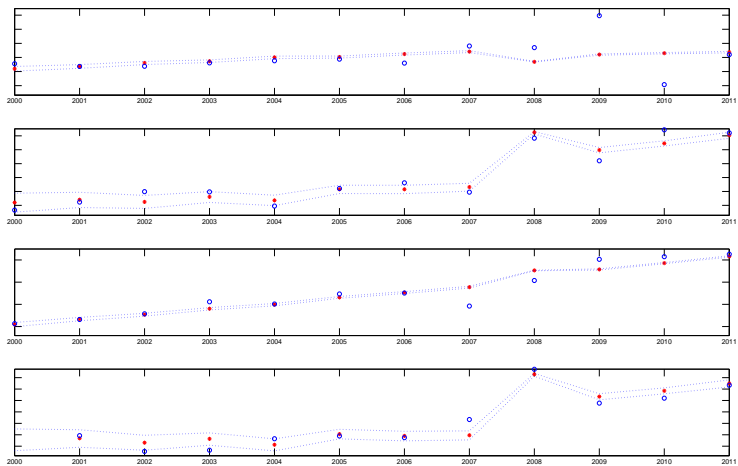,height=0.6\linewidth, width=\linewidth, angle=0}
\vspace{-25pt}
\caption{Estimates of $\nu_{1:n}$ (stars) with confidence bands (dashed). Estimates from \cite{ADL13} (circles) for comparison.}
\label{fig:termination_4par_params}
\end{center}
\end{figure}

\begin{figure}[!ht]
\begin{center}
\epsfig{file=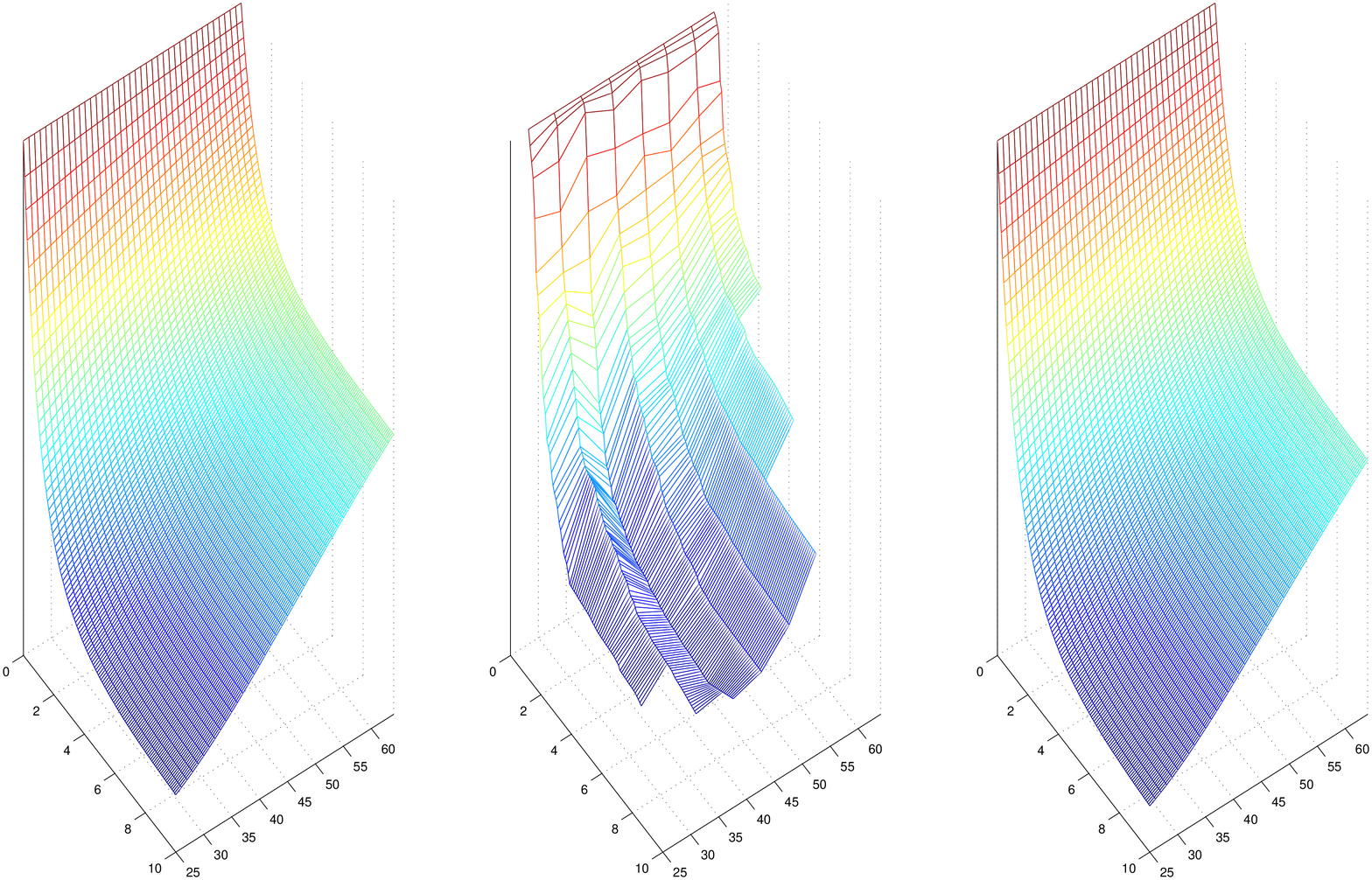,height=0.6\linewidth, width=\linewidth, angle=0}
\vspace{-25pt}
\caption{Left: Fitted termination surface, females 2006. Center: Kaplan-Meier. Right: Estimates from \cite{ADL13}.}
\label{fig:termination_4par_mesh_f2006}
\end{center}
\end{figure}

\begin{figure}[!ht]
\begin{center}
\epsfig{file=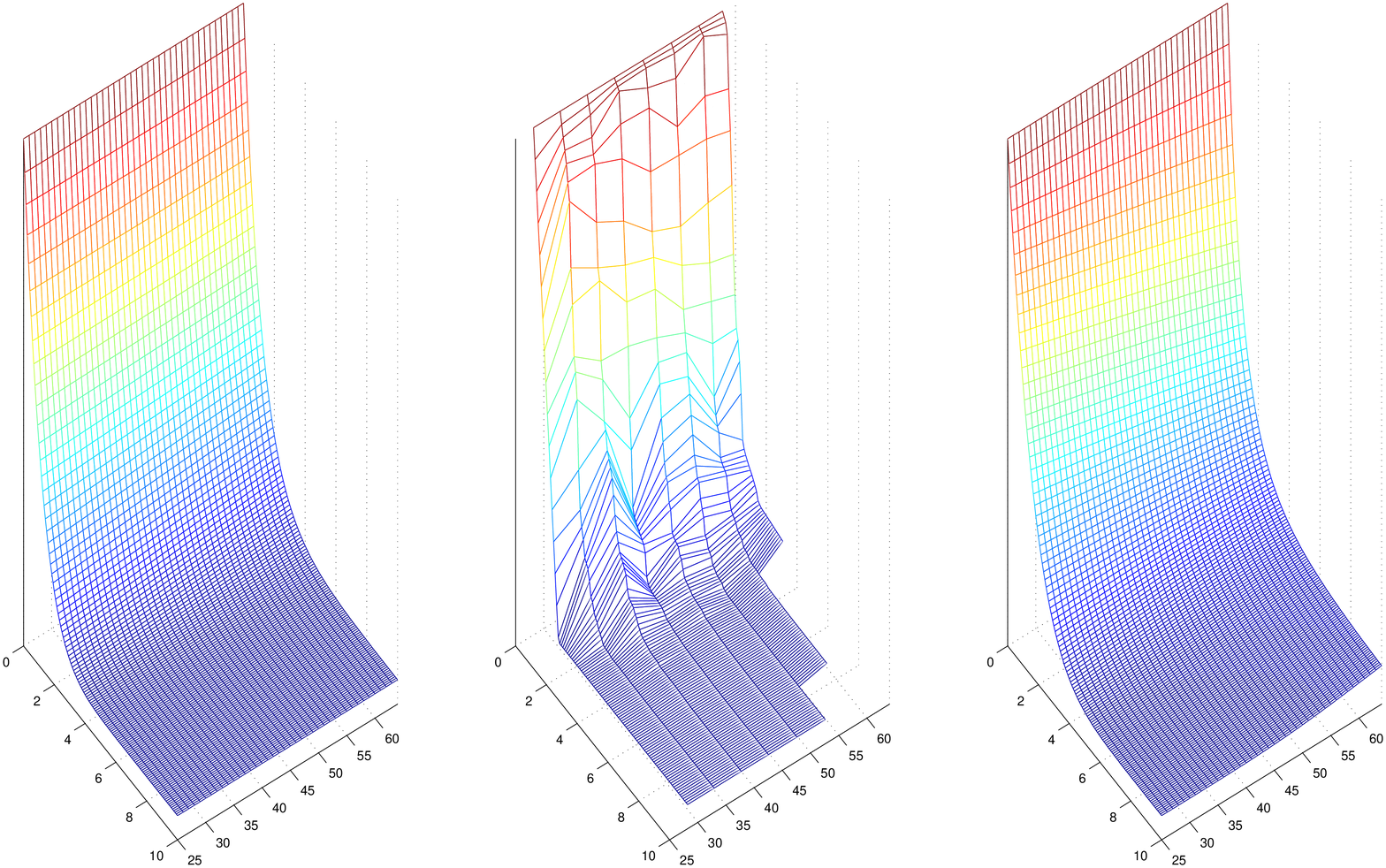,height=0.6\linewidth, width=\linewidth, angle=0}
\vspace{-25pt}
\caption{Left: Fitted termination surface, females 2010. Center: Kaplan-Meier. Right: Estimates from \cite{ADL13}.}
\label{fig:termination_4par_mesh_f2010}
\end{center}
\end{figure}

\begin{table}[!ht]
\caption{Relative difference of the estimated drift and volatility parameters between the two models.}
\label{table:4par}
\footnotesize
\begin{center}
\begin{tabular}{c|ccc}
% $\Delta z = \Delta t$
 & $\mu$ & $\sigma$ \\
\hline
$\nu^1$ & 1.82 & 0.15 \\
$\nu^2$ & 0.88 & 0.75 \\
$\nu^3$ & 0.98 & 0.33 \\
$\nu^4$ & 1.08 & 0.90 \\
\end{tabular} \\
\end{center}
%{\itshape{Notes.}} 
%Possible notes.
\end{table}
The HMM seems to provide estimates of $\nu_{1:n}$ and $\mu$ that are quite close to the estimates from the multi-period model, at least for $\nu^2$, $\nu^3$ and $\nu^4$. For $\nu^1$, the differences are quite pronounced from the year 2008 and onwards. Even so, the estimated termination surfaces from the HMM are close to the multi-period model surfaces for both 2006 and 2010. Further, we see that the estimated volatilities of $\nu^1,\ldots,\ \nu^4$ are lower for the HMM compared to the multi-period model, which is both reasonable and expected. The relatively large deviance for $\mu_1$ can be explained by the fact that the corresponding estimate for the multi-period model is close to zero, so the absolute difference is small.

\subsection{Six-factor model}
In order to refine the model we extend it to the following six-factor model:
\begin{eqnarray*}
\phi^1(x) & = & \frac{64-x}{39},\\
\phi^2(x) & = & \frac{x-25}{39}, \\
\psi^1(d) & = & 1,\\
\psi^2(d) & = & e^{-d},\\
\psi^3(d) & = & e^{-2d}.\\
\end{eqnarray*}
\noindent
The EM-algorithm stabilizes to within Monte Carlo error after about 140 iterations. The estimated parameters are displayed in Figure \ref{fig:termination_6par_params}. Figures \ref{fig:termination_6par_mesh_f2006}-\ref{fig:termination_6par_mesh_f2010} display the termination surface from the model alongside the corresponding Kaplan-Meier curves for the years 2006 and 2010, respectively. Results from \cite{ADL13} are displayed for reference. Table \ref{table:6par} displays the estimated drift and volatility parameters from the HMM as a fraction of the corresponding estimates from the multi-period model. 
\begin{figure}[!ht]
\begin{center}
\epsfig{file=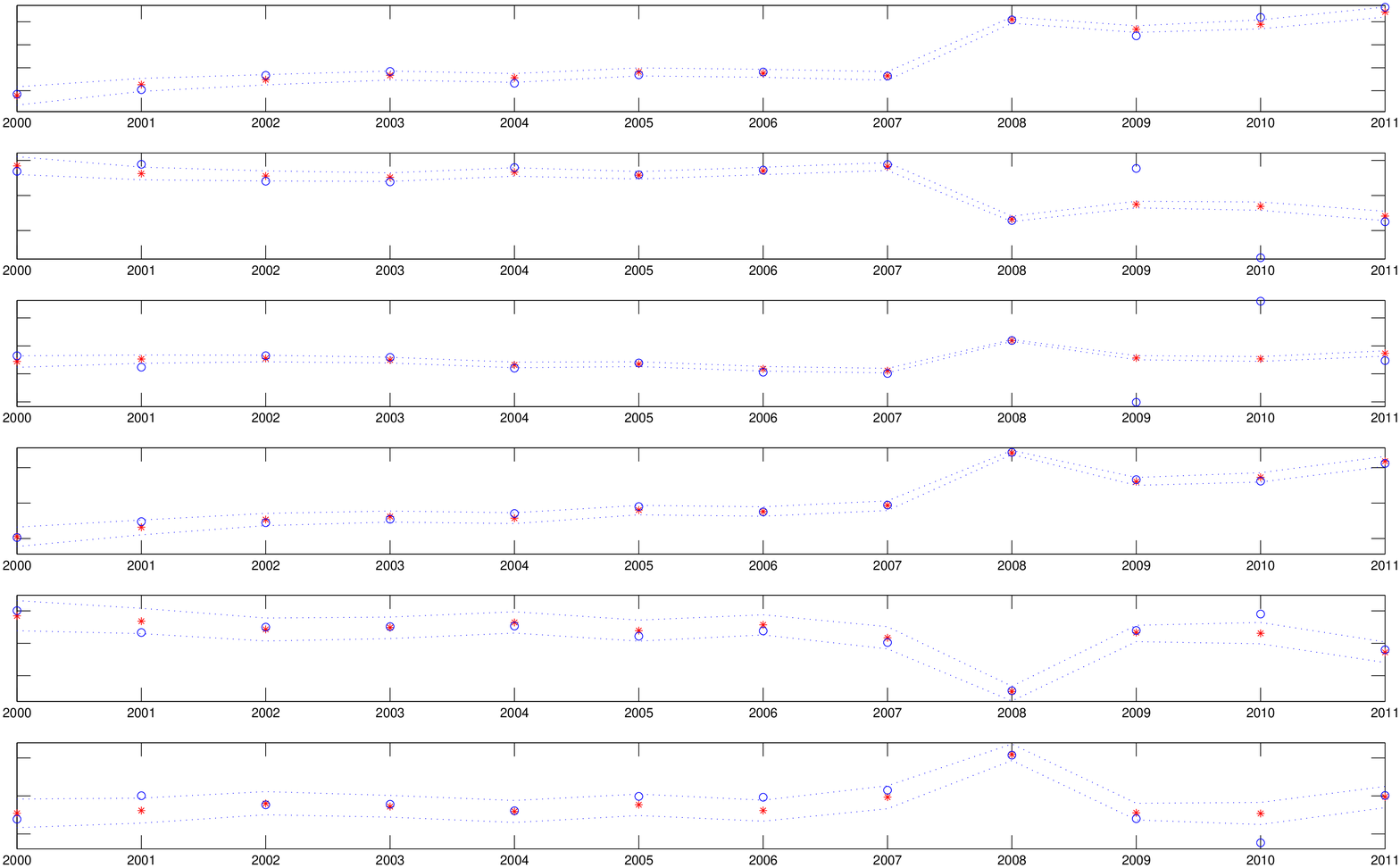,height=0.6\linewidth, width=\linewidth, angle=0}
\vspace{-25pt}
\caption{Estimates of $\nu_{1:n}$ (stars) with confidence bands (dashed). Estimates from \cite{ADL13} (circles) for comparison.}
\label{fig:termination_6par_params}
\end{center}
\end{figure}

\begin{figure}[!ht]
\begin{center}
\epsfig{file=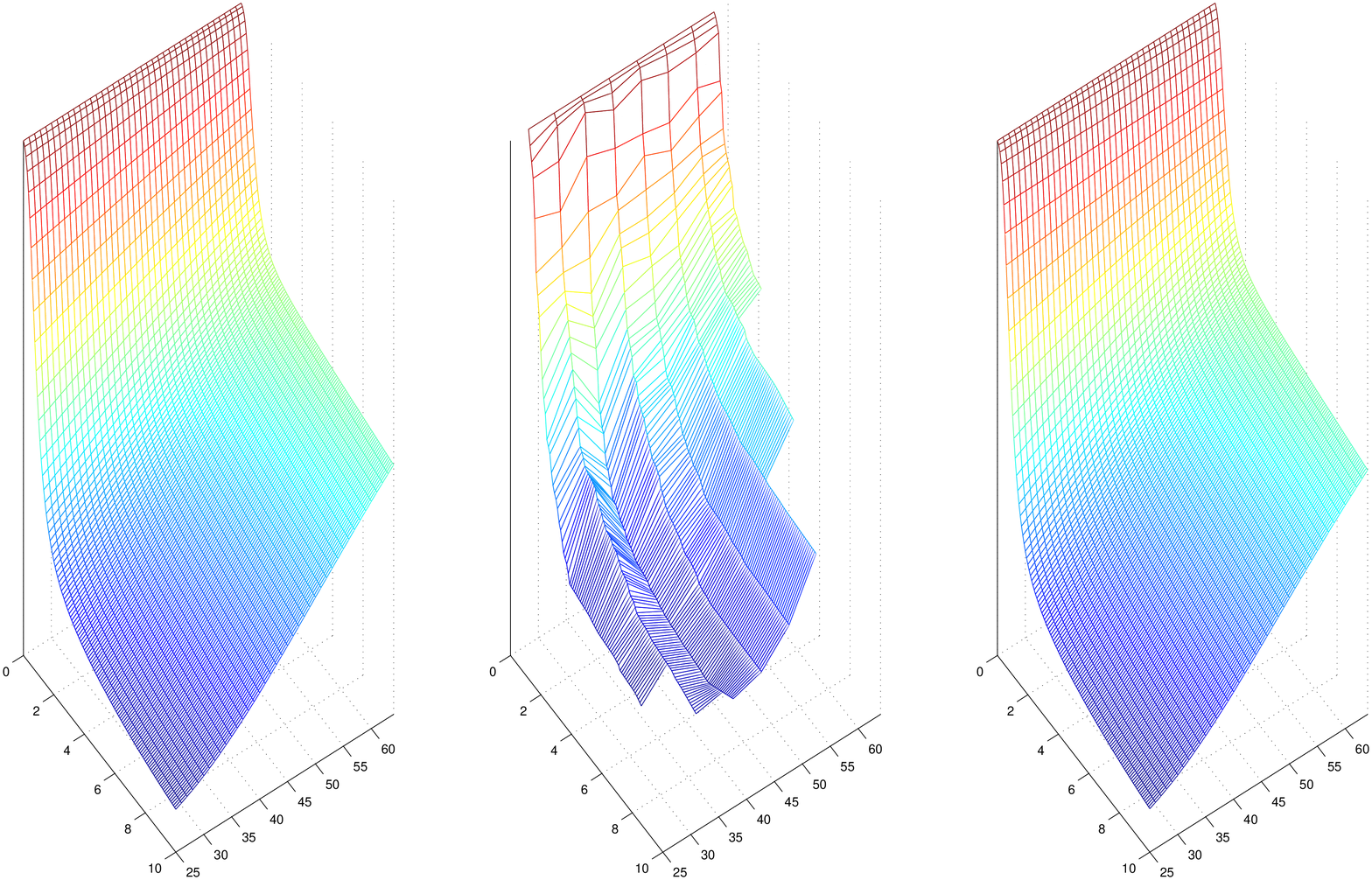,height=0.6\linewidth, width=\linewidth, angle=0}
\vspace{-25pt}
\caption{Left: Fitted termination surface, females 2006. Center: Kaplan-Meier. Right: Estimates from \cite{ADL13}.}
\label{fig:termination_6par_mesh_f2006}
\end{center}
\end{figure}

\begin{figure}[!ht]
\begin{center}
\epsfig{file=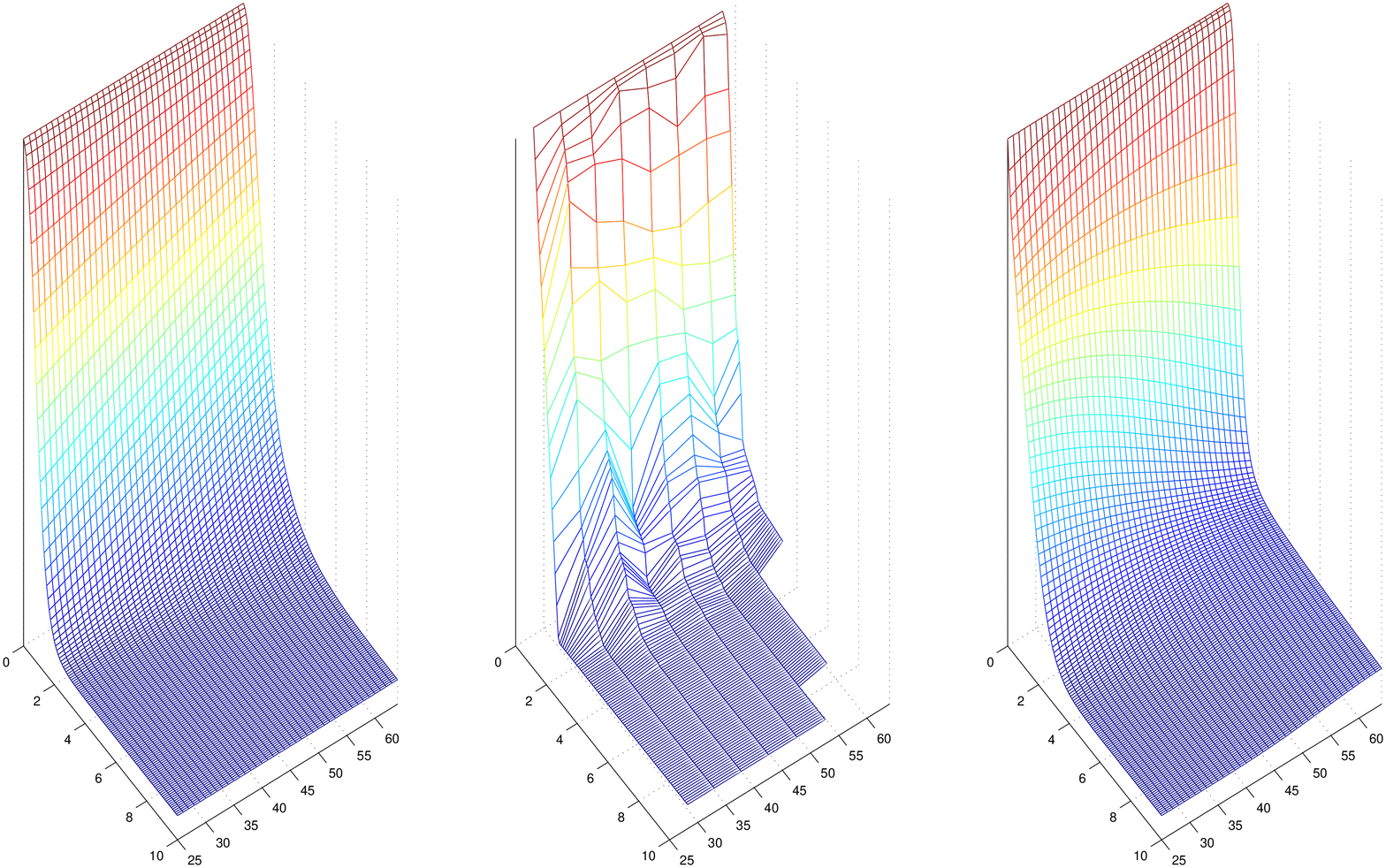,height=0.6\linewidth, width=\linewidth, angle=0}
\vspace{-25pt}
\caption{Left: Fitted termination surface, females 2010. Center: Kaplan-Meier. Right: Estimates from \cite{ADL13}.}
\label{fig:termination_6par_mesh_f2010}
\end{center}
\end{figure}

\begin{table}[!ht]
\caption{Relative difference of the estimated drift and volatility parameters between the two models.}
\label{table:6par}
\footnotesize
\begin{center}
\begin{tabular}{c|ccc}
% $\Delta z = \Delta t$
 & $\mu$ & $\sigma$ \\
\hline
$\nu^1$ & 0.96 & 0.84 \\
$\nu^2$ & 0.99 & 0.41 \\
$\nu^3$ & -1.56 & 0.24 \\
$\nu^4$ & 1.01 & 0.89 \\
$\nu^5$ & 0.91 & 0.85 \\
$\nu^6$ & 0.66 & 0.74 \\
\end{tabular} \\
\end{center}
%{\itshape{Notes.}} 
%Possible notes.
\end{table}
As for the four-factor model, we see that the estimated volatilities of $\nu^1,\ldots,\ \nu^6$ are lower for the HMM compared to the multi-period model. The estimates of $\nu_{1:n}$ and $\mu$ are similar to the estimates from the multi-period model, with some differences appearing from 2009 and onwards. Still, the estimated termination surfaces from the HMM are again close to the multi-period model surfaces for both 2006 and 2010, even though the estimated parameters for 2010 differ significantly between the two models. Since this was also the case with the four-factor model, it suggests a certain degree of robustness for the fitting procedure. The relatively large deviance for $\mu_3$, which even shows a change of sign, can be explained by the fact that the corresponding estimate for the multi-period model is close to zero, so that the absolute difference is small.

As a final note, we briefly comment on the nature of the unobservable environment process. The reform of the Swedish health insurance system in 2008 introduced harsher rules for obtaining benefits. On the other hand, for the proposed reform of 2014 the rules for obtaining benefits will become more lenient. Moreover, it is also possible that the population disability pattern follow the macro-economic trends of society in one way or another. All together, we are led to believe that the environment process follows a certain mean-reverting pattern. Unfortunately, the the data set used in this paper covers the rather short time period from 2000-2011, during which it may be hard to observe any mean-reverting pattern of the process. A topic for future research is to collect data from the new post-2014 regime and investigate whether mean-reverting processes such as the multivariate Vasicek model can be used to model the environment process.

\section{Acknowledgements}
The first author gratefully acknowledges financial support from the Swedish Export Credit Corp. (SEK). The second author gratefully acknowledges financial support from the Filip Lundberg  and Eir's 50 Years foundations.

\clearpage
\newpage
\bibliography{hmm}

\begin{thebibliography}{10}

\bibitem{ADL13}
{\sc Aro, H., Djehiche, B., and L\"ofdahl, B.}
\newblock Stochastic modelling of disability insurance in a multi-period
  framework.
\newblock {\em Scandinavian Actuarial Journal\/} (2013).

\bibitem{AroPennanen}
{\sc Aro, H., and Pennanen, T.}
\newblock A user-friendly approach to stochastic mortality modelling.
\newblock {\em European Actuarial Journal 1}, 2 (2011), 151--167.

\bibitem{brouhns2002poisson}
{\sc Brouhns, N., Denuit, M., and Vermunt, J.~K.}
\newblock A poisson log-bilinear regression approach to the construction of
  projected lifetables.
\newblock {\em Insurance: Mathematics and Economics 31}, 3 (2002), 373--393.

\bibitem{Christiansen}
{\sc Christiansen, M., Denuit, M., and Lazar, D.}
\newblock The {S}olvency {II} square-root formula for systematic biometric
  risk.
\newblock {\em Insurance: Mathematics and Economics 50\/} (2012), 257--265.

\bibitem{czado2005bayesian}
{\sc Czado, C., Delwarde, A., and Denuit, M.}
\newblock Bayesian {P}oisson log-bilinear mortality projections.
\newblock {\em Insurance: Mathematics and Economics 36}, 3 (2005), 260--284.

\bibitem{dempster1977maximum}
{\sc Dempster, A., Laird, N., and Rubin, D.}
\newblock Maximum likelihood from incomplete data via the {EM} algorithm.
\newblock {\em Journal of the Royal Statistical Society. Series B
  (Methodological)\/} (1977), 1--38.

\bibitem{DjehicheLofdahl14}
{\sc Djehiche, B., and L\"ofdahl, B.}
\newblock Risk aggregation and stochastic claims reserving in disability
  insurance.
\newblock {\em Insurance: Mathematics and Economics 59\/} (2014), 100--108.

\bibitem{duffie2009}
{\sc Duffie, D., Eckner, A., Horel, G., and Saita, L.}
\newblock Frailty correlated default.
\newblock {\em The Journal of Finance 64}, 5 (2009), 2089--2123.

\bibitem{Feyel}
{\sc Feyel, D., and \"Ust\"unel, A.}
\newblock Log-concave measures.
\newblock {\em {TWMS Journal of Pure and Applied Mathematics} 1}, 1 (2010),
  92--105.

\bibitem{HyndmanUllah}
{\sc Hyndman, R., and Ullah, M.~S.}
\newblock Robust forecasting of mortality and fertility rates: {A} functional
  data approach.
\newblock {\em Computational Statistics and Data Analysis 51\/} (2007),
  4942--4956.

\bibitem{leecarter1992}
{\sc Lee, R.~D., and Carter, L.~R.}
\newblock Modeling and forecasting us mortality.
\newblock {\em Journal of the American statistical association 87}, 419 (1992),
  659--671.

\bibitem{RenshawHaberman2000}
{\sc Renshaw, A., and Haberman, S.}
\newblock Modelling recent time trends in uk permanent health insurance
  recovery, mortality and claim inception transition intensities.
\newblock {\em Insurance: Mathematics and Economics 27\/} (2000), 365--396.

\bibitem{westerbornolsson2014}
{\sc Westerborn, J., and Olsson, J.}
\newblock Efficient particle-based online smoothing in general hidden {M}arkov
  models.
\newblock In {\em Acoustics, Speech and Signal Processing (ICASSP), 2014 IEEE
  International Conference on\/} (2014), IEEE, pp.~8003--8007.

\bibitem{wu1983convergence}
{\sc Wu, C.~J.}
\newblock On the convergence properties of the {EM} algorithm.
\newblock {\em The Annals of statistics\/} (1983), 95--103.

\end{thebibliography}
\bibliographystyle{acm}

\end{document}